\documentclass[11pt]{amsart}
\usepackage[a4paper,margin=1in]{geometry}               

\usepackage{amsmath}
\usepackage{amsthm}
\usepackage{mathtools} 
\usepackage{algorithm}
\usepackage{algorithmic}
\usepackage{enumitem}

\usepackage{tikz}
\usetikzlibrary{arrows,arrows.meta,backgrounds,calc,fit,decorations.pathreplacing,decorations.markings,shapes.geometric}
\def\borderColor{blue!60}

\tikzstyle{internal} = [draw, fill, shape=circle]
\tikzstyle{external} = [shape=circle]
\tikzstyle{square}   = [draw, fill, rectangle]
\tikzstyle{triangle} = [draw, fill, regular polygon, regular polygon sides=3, inner sep=3pt]
\tikzstyle{pentagon} = [draw, fill, regular polygon, regular polygon sides=5, inner sep=2pt, minimum size=14pt]

\usepackage[margin=1cm]{caption} 
\usepackage{subfig}

\usepackage[thinlines]{easytable}

\tikzset{every fit/.append style=text badly centered}

\usepackage[bookmarks=true,hypertexnames=false,pagebackref]{hyperref}
\hypersetup{colorlinks=true, citecolor=blue, linkcolor=red, urlcolor=blue}

\usepackage{scrtime}
\usepackage{ifthen}

\usepackage{cleveref}

\usepackage{todonotes}
\usepackage{mleftright}
\usetikzlibrary{positioning,chains,fit,shapes,calc}
\usetikzlibrary{trees}
\usetikzlibrary{decorations.pathreplacing}
\usetikzlibrary{decorations.pathmorphing}
\usetikzlibrary{decorations.markings}
\tikzset{>=latex} 

\usepackage{cool}
\Style{DSymb={\mathrm d},DShorten=true,IntegrateDifferentialDSymb=\mathrm{d}}

\newcommand{\numP}{\#{\bf P}}
\newcommand{\NP}{{\bf NP}}

\newcommand{\prs}{partial rejection sampling} 

\newcommand{\Ex}{\mathop{\mathbb{{}E}}\nolimits}

\def\*#1{\mathbf{#1}}
\def\+#1{\mathcal{#1}}
\def\-#1{\mathrm{#1}}

\newcommand{\abs}[1]{\left\vert#1\right\vert}
\newcommand{\ceil}[1]{\lceil#1\rceil}

\newcommand{\eps}{\varepsilon}

\renewcommand{\Pr}{\mathop{\mathrm{Pr}}\nolimits}

\newcommand{\Zrel}{Z_{rel}}
\newcommand{\Zreach}{Z_{reach}}

\newcommand{\outedge}{\ensuremath{A_{\mathrm{out}}}}
\newcommand{\reverse}[1]{\ensuremath{\overline{#1}}}
\newcommand{\sfix}{\ensuremath{S_{\fix}}}
\newcommand{\reach}{\textsc{Reachability}}
\newcommand{\relia}{\textsc{Reliability}}
\newcommand{\bireach}{\textsc{Bi-directed Reachability}}
\newcommand{\wt}{\mathrm{wt}}
\newcommand{\fix}{\mathrm{fix}}
\newcommand{\flip}{\mathrm{flip}}
\newcommand{\maxx}{\mathrm{max}}

\newtheorem{theorem}{Theorem}

\newtheorem{lemma}[theorem]{Lemma}
\newtheorem{claim}[theorem]{Claim}

\newtheorem{proposition}[theorem]{Proposition}

\newtheorem{definition}[theorem]{Definition}

\crefname{theorem}{Theorem}{Theorems}
\crefname{observation}{Observation}{Observations}
\crefname{claim}{Claim}{Claims}
\crefname{condition}{Condition}{Conditions}
\crefname{algorithm}{Algorithm}{Algorithms}
\crefname{property}{Property}{Properties}
\crefname{example}{Example}{Examples}
\crefname{fact}{Fact}{Facts}
\crefname{lemma}{Lemma}{Lemmas}
\crefname{corollary}{Corollary}{Corollaries}
\crefname{definition}{Definition}{Definitions}
\crefname{remark}{Remark}{Remarks}
\crefname{proposition}{Proposition}{Propositions}
\crefname{equation}{equation}{equations}

\makeatletter
\def\prob#1#2#3{\goodbreak\begin{list}{}{\labelwidth\z@ \itemindent-\leftmargin
                        \itemsep\z@  \topsep6\p@\@plus6\p@
                        \let\makelabel\descriptionlabel}
                      \item[\textbf{Name}]#1
                      \item[\textbf{Instance}]#2
                      \item[\textbf{Output}]#3
                \end{list}}
\makeatother

\title[Approximating network reliability]{A polynomial-time approximation algorithm for all-terminal network reliability}\thanks{The work described here was supported by the EPSRC research grant
EP/N004221/1 ``Algorithms that Count''.}

\author[H.\ Guo]{Heng Guo}
\address[Heng Guo]{School of Informatics, University of Edinburgh, Informatics Forum, Edinburgh, EH8 9AB, United Kingdom.}
\email{hguo@inf.ed.ac.uk}

\author[M.\ Jerrum]{Mark Jerrum}
\address[Mark Jerrum]{School of Mathematical Sciences,
Queen Mary, University of London, Mile End Road, London, E1 4NS, United Kingdom.}
\email{m.jerrum@qmul.ac.uk}

\begin{document}

\begin{abstract}
  We give a fully polynomial-time randomized approximation scheme (FPRAS) for the all-terminal network reliability problem,
  which is to determine the probability that, in a undirected graph, assuming each edge fails independently, the remaining graph is still connected.
  Our main contribution is to confirm a conjecture by Gorodezky and Pak ({\it Random Struct.\ Algorithms}, 2014),
  that the expected running time of the ``cluster-popping'' algorithm in bi-directed graphs is bounded by a polynomial in the size of the input.
\end{abstract}

\maketitle

\section{Introduction}

Network reliability problems are extensively studied \numP-hard problems \cite{Col87} (see also \cite{BP83, PB83, KL85, Ball86}).
In fact, these problems are amongst the first of those shown to be \numP-hard, and the two-terminal version is listed in Valiant's original thirteen \cite{Val79}.
The general setup is that in a given (undirected or directed) graph, every edge (or arc) $e$ has an independent probability~$p_e$ to fail,
and we are interested in various kinds of connectivity notions of the remaining graph.
For example, the two-terminal connectedness \cite{Val79} asks for the probability that for two vertices $s$ and $t$,
$s$ is connected to~$t$ in the remaining graph,
and the (undirected) all-terminal network reliability asks for the probability of all vertices being connected after edges fail.
The latter can also be viewed as a specialization of the Tutte polynomial $T_{G}(x,y)$ with $x=1$ and $y>1$,
yet another classic topic whose computational complexity is extensively studied \cite{JVW90, VW92, GJ08, GJ14}.

Prior to our work, the approximation complexity of network reliability problems remained elusive despite their importance.
There is no known efficient approximation algorithm (for any variant), but nor is there any evidence that such an algorithm does not exist.
A notable exception is Karger's fully polynomial-time randomized approximation scheme (FPRAS) for (undirected) all-terminal network \emph{unreliability} \cite{Kar99} 
(see also \cite{HS14, Kar16, Kar17} for more recent developments).
Although approximating unreliability is potentially more useful in practice,
it does not entail an approximation of its complement.

In this paper, we give an FPRAS for the all-terminal network reliability problem, defined below and denoted \relia.
\prob{\relia}{A (undirected) graph $G=(V,E)$, and failure probabilities $\mathbf{p}=(p_e)_{e\in E}$.}{$\Zrel(G;\mathbf{p})$, which is the probability that if each edge $e$ fails with probability $p_e$, the remaining graph is connected.}
When $p_e$ is independent of~$e$, \relia{} is an evaluation of the Tutte polynomial.  The \emph{Tutte polynomial}
is a two-variable polynomial $T_{G}(x,y)$ associated with a graph~$G$, which encodes much interesting information about~$G$.
As $(x,y)$ ranges over $\mathbb{R}^2$ or $\mathbb{C}^2$ we obtain a family of graph parameters, the so-called \emph{Tutte plane}.
As already noted, the study of the computational complexity of these parameters has a long history.
\relia{} with a uniform failure probability $0<p<1$ is equivalent to evaluating the Tutte polynomial $T_{G}(x,y)$ on the line $x=1$ and $y=\frac{1}{p}>1$.
Our algorithm is the first positive result on the complexity of the Tutte plane since Jerrum and Sinclair 
presented an FPRAS for the partition function of the ferromagnetic Ising model,
which is equivalent to the Tutte polynomial on the positive branch of the hyperbola $(x-1)(y-1)=2$ \cite{JS93}.  
It also answers a well-known open problem from 1980s, 
when the \numP-hardness of \relia\ was established \cite{Jer81,PB83} and the study of approximate counting initiated.
This problem is explicitly proposed in, for example, \cite[Conjecture 8.7.11]{Wel93} and \cite{Kar99}.
We note that many conjectures by Welsh (\cite[Chapter 8.7]{Wel93} and \cite{Wel99}) remain open,
and we hope that our work is only a beginning to answering these questions.

Another related and important reliability measure is \emph{reachability}, introduced and studied by Ball and Provan \cite{BP83}.
A directed graph $G=(V,A)$ with a distinguished root $r$ is said to be \emph{root-connected} if all vertices can reach $r$.
Reachability, denoted $\Zreach(G,r;\mathbf{p})$ for failure probabilities $\mathbf{p}=(p_e)_{e\in A}$, is the probability that, 
if each arc $e$ fails with probability $p_e$ independently, the remaining graph is still root-connected.  

We define the computational problem formally.
\prob{\reach}{A directed graph $G=(V,A)$ with root $r$, and failure probabilities $\mathbf{p}=(p_e)_{e\in A}$.}{$\Zreach(G,r;\mathbf{p})$.}
Exact polynomial-time algorithms are known when the graph is acyclic \cite{BP83} or has a small number of cycles \cite{Hag91}.
However, in general the problem is \numP-hard \cite{PB83}.

Ball \cite{Ball80} showed that \relia\ is equivalent to \reach\ in bi-directed graphs.
A bi-directed\footnote{There are other definitions of ``bi-directed graphs'' in the literature. Our definition is sometimes also called a symmetric directed graph.} graph is one 
where every arc has an anti-parallel twin with the same failure probability.
It is shown \cite{Ball80} that $Z_{rel}(G;\mathbf{p})=\Zreach(\overrightarrow{G},r;\mathbf{p}')$,
where $\overrightarrow{G}$ and $\mathbf{p}'$ are obtained by replacing every undirected edge in $G$ with a pair of anti-parallel arcs having the same failure probability in either direction,
and $r$ is chosen arbitrarily.  See \cref{lem:coupling}.

Our FPRAS for \relia\ utilizes this equivalence via approximating \reach\ in bi-directed graphs.
The core ingredient is the ``cluster-popping'' algorithm introduced by Gorodezky and Pak \cite{GP14}.
The goal is to sample root-connected subgraphs with probability proportional to their weights,
and then the reduction from counting to sampling is via a sequence of contractions.
A cluster is a subset of vertices not including the root and without any out-going arc.
The sampling algorithm randomizes all arcs independently, 
and then repeatedly resamples arcs going out from minimal clusters until no cluster is left,
at which point the remaining subgraph is guaranteed to be root-connected.
This approach is similar to Wilson's ``cycle-popping'' algorithm \cite{Wilson96} for rooted spanning trees,
and to the ``sink-popping'' algorithm \cite{CPP02} for sink-free orientations.
Gorodezky and Pak \cite{GP14} have noted that cluster-popping can take exponential time in general,
but they conjectured that in bi-directed graphs, 
the algorithm runs within polynomial-time.

We confirm this conjecture.
Let $p_{\maxx}$ be the maximum failure probability of edges (or arcs).
Let $m$ be the number of edges (or arcs) and $n$ the number of vertices.

\begin{theorem}  \label{thm:main}
  There is an FPRAS for \relia\ (or equivalently, \reach\ in bi-directed graphs).
  The expected running time is $O\left(\eps^{-2}(1-p_{\maxx})^{-3}m^2n^3\right)$ for an $(1\pm\eps)$-approximation.
  There is also an exact sampler to draw (edge-weighted) connected subgraphs with expected running time 
  $O\big((1-p_{\maxx})^{-1}m^2n\big)$.
\end{theorem}

We analyze the cluster-popping algorithm \cite{GP14} under the \emph{\prs} framework \cite{GJL17},
which is a general approach to sampling from a product distribution conditioned on avoiding a number of ``bad'' events. 
Partial rejection sampling is inspired by the Moser-Tardos algorithm for the Lov\'asz Local Lemma \cite{MT10}.
It starts with randomizing all variables independently, and then gradually eliminating ``bad'' events.
At every step, we need to find an appropriate set of variables to resample.
We call an instance \emph{extremal} \cite{KS11,Shearer85}, if any two bad events are either disjoint or independent.
For extremal instances, the resampling set can be simply chosen to be the set of all variables involved in occurring bad events \cite{GJL17},
and the algorithm becomes exactly the same as the Moser-Tardos resampling algorithm \cite{MT10}.
In particular, all three ``popping'' algorithms \cite{Wilson96, CPP02, GP14} are special cases of \prs\ for extremal instances.
In case of cluster-popping, the bad events are exactly minimal clusters.

The advantage of the \prs\ treatment is that we have an explicit formula for the expected number of resampling events for any extremal instance \cite{KS11,GJL17},
which equals to the ratio between the probability of having exactly one bad event and the probability of avoiding all bad events.
In order to bound this ratio, we use a combinatorial encoding idea and design a mapping from subgraphs with a unique minimal cluster to root-connected subgraphs.
To make this mapping injective, we record an extra vertex and an arc so that we can recover the pre-image.
This extra cost is upper-bounded by a polynomial in the size of the graph.


Cluster-popping only draws root-connected subgraphs in the bi-directed setting.
In order to sample connected subgraphs in the undirected setting,
we provide an alternative proof of the equivalence between \relia\ and \reach\ in bi-directed graphs,
which essentially is a coupling argument.
This coupling has a new consequence that, once we have a sample of a root-connected subgraph,
it is easy to generate a connected subgraph according to the correct distribution.

In \Cref{sec:prelim} we introduce the cluster-popping algorithm and the partial rejection framework.
In \Cref{sec:run-time} we analyze its running time in bi-directed graphs.
For completeness, in \Cref{sec:counting} we include the approximate counting algorithm due to Gorodezky and Pak \cite{GP14}.
In \Cref{sec:equivalence} we give a coupling proof of the equivalence between \relia\ and \reach\ in bi-directed graphs.
In \Cref{sec:fixedsize} we use our sampling algorithm to show how to approximately count the number of connected subgraphs of a fixed size.
In \Cref{sec:conclusion} we conclude by mentioning a few open problems.

\section{Cluster-popping}\label{sec:prelim}

Let $G=(V,A)$ be a directed\footnote{It is easy to see that in a undirected graph, 
reachability is the same as all-terminal reliability.} graph with root $r$.
The graph $G$ is called \emph{root-connected} if there is a directed path in $G$ from every non-root vertex to~$r$.
Let $0<p_e<1$ be the failure probability of arc $e$,
and define the weight of a subgraph $S$ to be $\wt(S):=\prod_{e\in S}(1-p_e)\prod_{e\not\in S} p_{e}$.
Then reachability, $\Zreach(G,r;\mathbf{p})$, is defined as follows,
\begin{align*}
  \Zreach(G,r;\mathbf{p}) := \sum_{\substack{S\subseteq A\\\text{$(V,S)$ is root-connected}}} \wt(S).
\end{align*}
Here, $\mathbf{p}=(p_e:e\in A)$ denotes the vector of failure probabilities.

Let $\pi_G(\cdot)$ (or $\pi(\cdot)$ for short) be the distribution resulting 
from choosing each arc~$e$ independently with probability $1-p_e$, 
and conditioning on the resulting graph being root-connected.
In other words, the support of $\pi(\cdot)$ is the collection of all root-connected subgraphs, 
and the probability of each subgraph $S$ is proportional to its weight $\wt(S)$.
Then $\Zreach(G,r;\mathbf{p})$ is the normalizing factor of the distribution $\pi(\cdot)$.
Gorodezky and Pak \cite{GP14} have shown that approximating $\Zreach(G,r;\mathbf{p})$ can be reduced to sampling from $\pi(\cdot)$ when the graph is bi-directed.

The cluster-popping algorithm of Gorodezky and Pak \cite{GP14}, to sample root-connected subgraphs from $\pi(\cdot)$, can be viewed as a special case of \prs\ \cite{GJL17} for extremal instances.
With every arc~$e$ of $G$ we associate a random variable that records whether that
arc has failed.  Bad events are characterized by the following notion of clusters.

\begin{definition}\label{def:cluster}
  In a directed graph $(V,A)$ with root $r$, a subset $C\subseteq V$ of vertices is called a \emph{cluster} 
  if $r\not\in C$ and there is no arc $u\rightarrow v \in A$ such that $u\in C$ and $v\not\in C$.

  We say $C$ is a \emph{minimal cluster} if $C$ is a cluster and for any proper subset $C'\subset C$, $C'$ is not a cluster.
\end{definition}

If $(V,A)$ contains no cluster, then it is root-connected.
For each vertex $v$, let $\outedge(v)$ be the set of outgoing arcs from $v$.
We also abuse the notation to write $\outedge(S) = \bigcup_{v\in S}\outedge(v)$ 
for a subset $S\subset V$ of vertices.
Notice that $\outedge(S)$ contains edges between vertices inside $S$.
To ``pop'' a cluster~$C$, we re-randomize all arcs in $\outedge(C)$.
However, re-randomizing clusters does not yield the desired distribution.
We will instead re-randomize minimal clusters.

\begin{claim}\label{clm:cluster:sc}
  Any minimal cluster is strongly connected.
\end{claim}
\begin{proof}
  Let $C$ be a minimal cluster, and $v\in C$ be an arbitrary vertex in $C$.
  We claim that $v$ can reach all vertices of $C$.
  If not, let $C'$ be the set of reachable vertices of $v$ and $C'\subsetneq C$.
  Since $C'$ does not have any outgoing arcs, $C'$ is a cluster.
  This contradicts to the minimality of $C$.
\end{proof}

\begin{claim}\label{clm:cluster:extremal}
  If $C_1$ and $C_2$ are two distinct minimal clusters,
  then $C_1\cap C_2 = \emptyset$.
\end{claim}
\begin{proof}
  By \cref{clm:cluster:sc}, $C_1$ and $C_2$ are both strongly connected components.
  If $C_1\cap C_2 \neq \emptyset$, then they must be identical.
\end{proof}

For every subset $C\subseteq V$ of vertices, we define a bad event $B_C$, which occurs if $C$ is a minimal cluster.
Observe that $B_C$ relies only on the status of arcs in $\outedge(C)$.
Thus, if $C_1 \cap C_2 = \emptyset$, then $B_{C_1}$ and $B_{C_2}$ are independent, even if some of their vertices are adjacent.
By \cref{clm:cluster:extremal}, we know that two bad events $B_{C_1}$ and $B_{C_2}$ are either independent or disjoint.
Thus the aforementioned extremal condition is met.
Moreover, it was shown \cite[Theorem 8]{GJL17} that if the instance is extremal, 
then at every step, we only need to resample variables involved in occurring bad events.
This leads to the cluster-popping algorithm of Gorodezky and Pak \cite{GP14}, which is formally described in \cref{alg:cluster-popping}.

\begin{algorithm}
  \caption{Cluster Popping}
  \label{alg:cluster-popping}
  \begin{algorithmic}
    \STATE {Let $S$ be a subset of arcs by choosing each arc $e$ with probability $1-p_e$ independently.}
    \WHILE {There is a cluster in $(V,S)$.}
    \STATE {Let $C_1,\dots,C_k$ be all minimal clusters in $(V,S)$, and $C=\bigcup_{i=1}^k C_i$.}
    \STATE {Re-randomize all arcs in $\outedge(C)$ to get a new $S$.}
    \ENDWHILE
    \RETURN {$S$}
  \end{algorithmic}
\end{algorithm}

The correctness of \Cref{alg:cluster-popping} is first shown by Gorodezky and Pak \cite{GP14}.
It can also be easily verified using \cite[Theorem 8]{GJL17}.

\begin{theorem}[\protect{\cite[Theorem 2.2]{GP14}}]\label{thm:correctness}
  The output of \Cref{alg:cluster-popping} is drawn from $\pi_G$.
\end{theorem}

An advantage of thinking in the \prs\ framework is that we have a closed form formula for the expected running time of these algorithms on extremal instances.
Let $\Omega_k$ be the collection of subgraphs with $k$ minimal clusters,
and 
\begin{align*}
  Z_k := \sum_{S\in\Omega_k} \wt(S).
\end{align*}
Then $Z_0 = \Zreach(G,r;\mathbf{p})$, since any subgraph in $\Omega_0$ has no cluster and is thus root-connected.

\begin{theorem}[\cite{GJL17}]  \label{thm:extremal:runningtime}
  Let $T$ be the number of resampled events of the partial rejection sampling algorithm for extremal instances.
  Then 
  \begin{align*}
    \Ex T = \frac{Z_1}{Z_0}.
  \end{align*}
  In particular, for \cref{alg:cluster-popping}, $T$ is the number of popped clusters.
\end{theorem}

\Cref{thm:extremal:runningtime} can be shown via manipulating generating functions.
The less-than-or-equal-to direction of \Cref{thm:extremal:runningtime} was shown by Kolipaka
and Szegedy~\cite{KS11}, which is the direction we will need later.
The other direction is useful to show running-time lower bounds,
but that is not our focus in this paper.

\section{Running time of \texorpdfstring{\cref{alg:cluster-popping}}{Algorithm 1} in bi-directed graphs}\label{sec:run-time}

Gorodezky and Pak \cite{GP14} have given examples of directed graphs in which \cref{alg:cluster-popping} requires exponential time.
In the following we focus on bi-directed graphs.
A graph $G$ is called \emph{bi-directed} if $u\rightarrow v$ is present in $G$, then $v\rightarrow u$ is present in $G$ as well,
and the failure probabilities are the same for these two arcs.
We use \bireach\ to denote \reach\ in bi-directed graphs.
For an arc $e=u\rightarrow v$, let $\reverse{e} := v\rightarrow u$ denote its reverse arc.
Then in a bi-directed graph, $p_e = p_{\reverse{e}}$.

\begin{lemma}  \label{lem:ratio}
  Let $G=(V,A)$ be a root-connected bi-directed graph with root $r$.
  We have that $Z_1\le \max_{e\in A} \left\{ \frac{p_e}{1-p_e} \right\}mn Z_0$, where $n = \abs{V}$, and $m = \abs{A}$.
\end{lemma}
\begin{proof}
  We construct an injective mapping $\varphi:\Omega_1\rightarrow \Omega_0\times V\times A$.
  For each subgraph $S\in\Omega_1$, $\varphi(S)$ is defined by ``repairing'' $S$ so that no minimal cluster is present.
  We choose in advance an arbitrary ordering of vertices and arcs.
  Let $C$ be the unique minimal cluster in $S$ and $v$ be the first vertex in $C$.
  Let $R$ denote the set of all vertices which can reach the root $r$ in the subgraph $S$.
  Since $S\in\Omega_1$, $R\neq V$.
  Let $U = V\setminus R$.
  Since $G$ is root-connected, there is an arc in $A$ from $U$ to $R$.
  Let $u \rightarrow u'$ be the first such arc, where $u\in U$ and $u'\in R$.
  We let
  \begin{align*}
    \varphi(S) := (\sfix,v,u\rightarrow u'),
  \end{align*}
  where $\sfix\in\Omega_0$ is defined next.
  \Cref{fig:illustration-Z1} is an illustration of these objects.

  \begin{figure}[htbp]
    \centering
      \begin{tikzpicture}[scale=0.7, inner sep=2pt, transform shape]
        \draw (0,0) ellipse (4 and 2) [thick];
        \draw (5,1.5) node {\LARGE $(V,S)$};
        \begin{scope}
          \clip (0,0) ellipse (4 and 2);
          \draw (3,0) ellipse (3 and 3);
        \end{scope}
        \draw (3,0.6) node [draw,fill,shape=circle,color=black, label=0:{\LARGE $r$}] {};
        \draw (1.8,-0.3) node {\LARGE $R$};
        \draw (-1.8,0) node {\LARGE $U$};

        \draw (-0.1,1.5) node [draw,fill,shape=circle,color=black, label=180:{\LARGE $u$}] (u) {};
        \draw  (0.8,1.1) node [draw,fill,shape=circle,color=black, label=0:{\LARGE $u'$}] (u'){};
        \path (u) edge [dashed, thick, ->] (u');
      \end{tikzpicture}  
    \caption{An illustration of $R$, $U$, and $u\rightarrow u'$.}
    \label{fig:illustration-Z1}
  \end{figure}
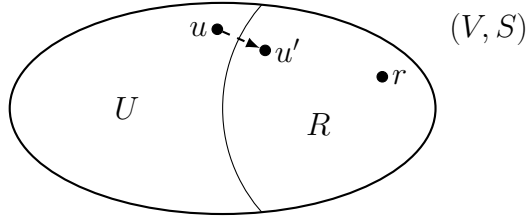

  Consider the subgraph $H=(U,S[U])$, where 
  \begin{align*}
    S[U] := \{x\rightarrow y\mid x\in U,\; y\in U,\; x\rightarrow y\in S\}.
  \end{align*}
  We consider the directed acyclic graph (DAG) of strongly connected components of $H$, and call it $\widehat{H}$.
  (We use the decoration $\widehat{\;}$ to denote arcs, vertices, etc.\ in $\widehat{H}$.)  
  To be more precise, we replace each strongly connected component by a single vertex.
  For a vertex $w\in U$, let $[w]$ denote the strongly connected component containing $w$.
  For example, $[v]$ is the same as the minimal cluster $C$ by \cref{clm:cluster:sc}.
  We may also view $[w]$ as a vertex in $\widehat{H}$ and we do not distinguish the two views.
  The arcs in $\widehat{H}$ are naturally induced by $S[U]$.
  Namely, for $[x]\neq[y]$, an arc $[x]\rightarrow [y]$ is present in $\widehat{H}$ if there exists $x'\in[x]$, $y'\in[y]$ such that $x'\rightarrow y' \in S$.

  We claim that $\widehat{H}$ is root-connected with root $[v]$. 
  This is because $[v]$ must be the unique sink in $\widehat{H}$ and $\widehat{H}$ is acyclic.
  If there is another sink $[w]$ where $v\not\in[w]$, then $[w]$ is a minimal cluster in $H$.
  This contradicts $S\in\Omega_1$.

  Since $\widehat{H}$ is root-connected, there is at least one path from $[u]$ to $[v]$.
  Let $\widehat{W}$ denote the set of vertices of $\widehat{H}$ that can be reached from $[u]$ in $\widehat{H}$ (including $[u]$),
  and $W:= \{x\mid [x]\in \widehat{W}\}$.
  Then $W$ is a cluster and $[u]$ is the unique source in $\widehat{H}[\widehat{W}]$.
  As $\widehat{H}$ is root-connected, $[v]\in\widehat{W}$.
  Define
  \begin{align*}
    S_{\flip} := \big\{x\rightarrow y\bigm| [x]\neq[y],\; x,y\in W,\; \text{and }x\rightarrow y \in S\big\},
  \end{align*}
  which is the set of edges to be flipped.
  Notice that $S[W]$ is different from $S_{\flip}$, namely all arcs that are inside strongly connected components are ignored in $S_{\flip}$.
  Now we are ready to define \sfix.
  We reverse all arcs in $S_{\flip}$ and add the arc $u\rightarrow u'$ to fix the minimal cluster.
  Formally, let
  \begin{align*}
    \sfix := S \cup \{u\rightarrow u'\} \cup \{ y\rightarrow x \mid x\rightarrow y\in S_{\flip}\} \setminus S_{\flip}.
  \end{align*}
  \Cref{fig:illustration-S-flip} is an example of these objects we defined.

  \begin{figure}[htbp]
    \centering
      \begin{tikzpicture}[scale=1, inner sep=2pt, transform shape]
        \draw (4,0)    node [draw,fill,shape=circle,color=black, label=0:{\LARGE $r$}] (r) {};
        \draw (2,1.5)  node [draw,fill,shape=circle,color=black, label=5:{\LARGE $u'$}] (u') {};
        \draw (2,0)    node [draw,fill,shape=circle,color=black] (v2) {};
        \draw (2,-1.5) node [draw,fill,shape=circle,color=black] (v3) {};
        \draw (u') edge [thick, ->] (r);
        \draw (u') edge [thick, ->] (v2);        
        \draw (v2) edge [thick, ->] (r);
        \draw (v3) edge [thick, ->, bend left] (r);
        \draw (r)  edge [thick, ->, bend left] (v3);

        \draw (1,0)  node [draw,fill,shape=circle,color=black] (w1) {};
        \draw (0,1)  node [draw,fill,shape=circle,color=black, label=90:{\LARGE $u$}] (u) {};
        \draw (0,-1) node [draw,fill,shape=circle,color=black] (w3) {};
        \draw (w1) edge [thick, ->] (u);
        \draw (u) edge [thick, ->] (w3);
        \draw (w3) edge [thick, ->] (w1);

        \draw (u) edge [thick, dashed, ->] (u');

        \draw (-2, 1) node [draw,fill,shape=circle,color=black] (x1) {};
        \draw (-2,-1) node [draw,fill,shape=circle,color=black, label=225:{\LARGE $v$}] (x2) {};

        \draw (x1) edge [thick, ->, bend left] (x2);
        \draw (x2) edge [thick, ->, bend left] (x1);

        \draw (u)  edge [thick, ->, bend right, color = red] (x1);
        \draw (x1) edge [thick, ->, bend right, dashed] (u);
        \draw (w3) edge [thick, ->, bend right, color = red] (x2);
        \draw (x2) edge [thick, ->, bend right, dashed] (w3); 

        \draw (4.9,1) node {\LARGE $R$};
        \draw (-3,0) node {\LARGE $U$};

        \draw (-2.5,0)    node (aux1) {};
        \draw (-1.5,1.5)  node (aux2) {};
        \draw (-1.5,-1.5) node (aux3) {};

        \draw (2,1.8) node (aux4) {};
        \draw (1.8,0) node (aux5) {};
        \draw (4.3,0) node (aux6) {};

        \begin{pgfonlayer}{background}
          \node[inner sep=4pt,transform shape=false,draw=\borderColor,thick,rounded corners,fit = (x1) (x2) (u) (w1) (w3) (aux1) (aux2) (aux3)] {};
          \node[inner sep=4pt,transform shape=false,draw=\borderColor,thick,rounded corners,fit = (r) (v2) (v3) (u') (aux4) (aux5) (aux6)] {};
        \end{pgfonlayer}
      \end{tikzpicture}  
      \caption{An example of $S_{\flip}$ (red arcs) in the subgraph $(V,S)$. 
      Dashed arcs are to be added to $\sfix$. 
      The underlying graph has more arcs than are drawn here.}
    \label{fig:illustration-S-flip}
  \end{figure}
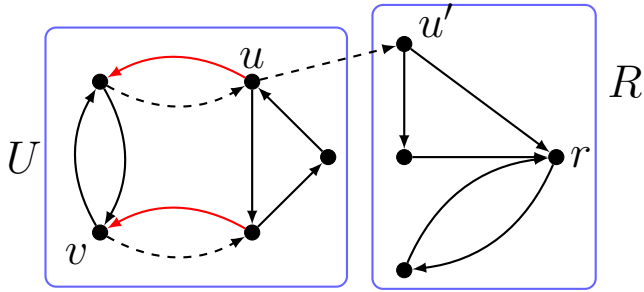

  Let $\widehat{H}_{\fix}$ be the graph obtained from $\widehat{H}$ by reversing all arcs induced by $S_{\flip}$.
  Observe that $[u]$ becomes the unique sink in $\widehat{H}_{\fix}[\widehat{W}]$ (and $[v]$ becomes the unique source).

  We verify that $\sfix\in\Omega_0$.
  For any $x\in R$, $x$ can still reach $r$ in $(V,\sfix)$ since the path from $x$ to $r$ in $(V,S)$ is not changed.
  Since $u\rightarrow u'\in\sfix$, $u$ can reach $u'\in R$ and hence $r$.
  For any $y\in W$, $y$ can reach $u$ as $[u]$ is the unique sink in $\widehat{H}_{\fix}[\widehat{W}]$.
  For any $z\in U\setminus W$, $z$ can reach $v\in W$ since the path from $z$ to $v$ in $(V,S)$ is not changed.

  Next we verify that $\varphi$ is injective.
  To do so, we show that we can recover $S$ given $\sfix$, $u\rightarrow u'$, and $v$.
  First remove $u\rightarrow u'$ from $\sfix$.
  The set of vertices which can reach $r$ in $(V,\sfix\setminus\{u\rightarrow u'\})$ is exactly $R$ in $(V,S)$.
  Namely we can recover $U$ and $R$.
  As a consequence, we can recover all arcs in $S$ that are incident with~$R$, as these arcs are not changed.

  What is left to do is to recover arcs in $S[U]$.
  To do so, we need to find out which arcs have been flipped.
  We claim that $\widehat{H}_{\fix}$ is acyclic.
  Suppose there is a cycle in $\widehat{H}_{\fix}$.
  Since $\widehat{H}$ is acyclic, the cycle must involve flipped arcs and thus vertices in $\widehat{W}$.
  Let $[x]\in\widehat{W}$ be the lowest one under the topological ordering of $\widehat{H}[\widehat{W}]$.
  Since $\widehat{W}$ is a cluster, the outgoing arc $[x]\rightarrow [y]$ along the cycle in $\widehat{H}_{\fix}$ must have been flipped, 
  implying that $[y]\in\widehat{W}$ and $[y]\rightarrow [x]$ is in $\widehat{H}[\widehat{W}]$.
  This contradicts to the minimality of $[x]$.

  Since $\widehat{H}_{\fix}$ is acyclic, the strongly connected components of $H_{\fix}:=(U,\sfix[U])$ are identical to those of $H=(U,S[U])$.
  Hence contracting all strongly connected components of $H_{\fix}$ results in exactly $\widehat{H}_{\fix}$.
  All we need to recover now is the set $\widehat{W}$.
  Let $\widehat{W}'$ be the set of vertices reachable from $[v]$ in $\widehat{H}_{\fix}$.
  It is easy to see that $\widehat{W} \subseteq \widehat{W}'$.
  We claim that actually $\widehat{W} = \widehat{W}'$.
  For any $[x]\in\widehat{W}'$,
  there is a path from $[v]$ to $[x]$ in $\widehat{H}_{\fix}$.
  Suppose $[x]\not\in\widehat{W}$.
  Since $[v]\in\widehat{W}$, we may assume that $[y]$ is the first vertex along the path such that $[y]\rightarrow[z]$ where $[z]\not\in\widehat{W}$.
  Thus $[y]\rightarrow[z]$ has not been flipped and is present in $\widehat{H}$.
  However, this contradicts the fact that $\widehat{W}$ is a cluster in $\widehat{H}$.

  To summarize, given $\sfix$, $u\rightarrow u'$, and $v$,
  we may uniquely recover $S$.
  Hence the mapping $\varphi$ is injective.
  Moreover, flipping arcs does not change the weight as $p_e = p_{\reverse{e}}$,
  and only adding the arc $u\rightarrow u'$ would.
  We have that $\wt(\sfix) = \frac{1-p_{u\rightarrow u'}}{p_{u\rightarrow u'}}\wt(S)$.
  The lemma follows.
\end{proof}

We remark that an alternative way of repairing $S$ in the proof above is to reverse all arcs in $S[W]$ without defining $S_{\flip}$.
The key point is that doing so leaves the strongly connected components intact.
However this makes the argument less intuitive.

Let $p_{\maxx} = \max_{e\in A}p_e$. Combining \cref{thm:extremal:runningtime} and \cref{lem:ratio}, we have the following theorem.
Notice that for each popping, we resample only a subset of arcs.

\begin{theorem}\label{thm:fpaus}  
  Let $T$ be the expected number of popped clusters in \cref{alg:cluster-popping}.
  For a root-connected bi-directed graph $G=(V,A)$, $\Ex T\le \frac{p_{\maxx}}{1-p_{\maxx}} mn$, where $n = \abs{V}$, and $m = \abs{A}$.
  The expected running time is asymptotically at most $\frac{p_{\maxx}}{1-p_{\maxx}} m^2n$.
\end{theorem}

\section{Approximate counting}\label{sec:counting}

We include the approximate counting algorithm of Gorodezky and Pak \cite{GP14} for completeness.
Let $G=(V,A)$ be an instance of \bireach\ with root $r$ and parameters~$\mathbf{p}$.
We construct a sequence of graphs $G_0,..,G_{n-1}$ where $n=\abs{V}$ and $G_0=G$.
Given $G_{i-1}$, choose two arbitrary adjacent vertices $u_i$ and $v_i$, 
remove all arcs between $u_i$ and $v_i$ (in either direction),
and identify $u_i$ and $v_i$ to get $G_i =(V_i,A_i)$.
Namely we contract all arcs between $u_i$ and $v_i$, 
but parallel arcs in the resulting graph are preserved.
If one of $u_i$ and $v_i$ is $r$, the new vertex is labelled $r$.
Thus $G_{n-1}=(\{r\},\emptyset)$.
Since $A_i$ is always a subset of $A$,
we denote by $\mathbf{p}_i$ the parameters $\mathbf{p}$ restricted to $A_i$.

For $i=1,\dots,n-1$, define a random variable $R_i$ as follows:
\begin{align*}
  R_i :=
  \begin{cases}
    1 & \text{$(V_{i-1},S_{i-1})$ is root-connected in $G_{i-1}$;}\\
    0 & \text{otherwise,}
  \end{cases}
\end{align*}
where $S_{i-1}\subset A_{i-1}$ is a random root-connected subgraph drawn from the distribution $\pi_{G_i}(\cdot)$,
together with all arcs $e$ between $u_i$ and $v_i$ added independently with probability $1-p_e$.
It is easy to see that
\begin{align*}
  \Ex R_i = \frac{\Zreach(G_{i-1},r;\mathbf{p}_{i-1})}{\Zreach(G_{i},r;\mathbf{p}_i)},
\end{align*}
and 
\begin{align*}
  \Zreach(G,r;\mathbf{p}) = \prod_{i=1}^{n-1} \Ex R_i.
\end{align*}

Let $p_{\maxx} = \max_{e\in A}p_e$ and $s = \ceil{5(1-p_{\maxx})^{-2}(n-1)\eps^{-2}}$ where $s$ is the desired precision.
We estimate $\Ex R_i$ by the empirical mean of $s$ independent samples of $Z_i$, denoted by $\widetilde{R}_i$,
and let $\widetilde{Z}=\prod_{i=1}^{n-1} \widetilde{R}_i$ and $Z=\Zreach(G,r;\mathbf{p})$.
Gorodezky and Pak \cite{GP14} showed the following.

\begin{proposition}[\protect{\cite[Section 9]{GP14}}]
  $\Pr\left(\abs{Z-\widetilde{Z}}>\eps Z\right)\leq 1/4$.
\end{proposition}

In order to sample $Z_i$, we use \cref{alg:cluster-popping} to draw independent samples of root-connected subgraphs.
\cref{thm:fpaus} implies that each sample takes at most $\frac{p_{\maxx}}{1-p_{\maxx}}m^2n$ time in expectation.
We need $O\left(\frac{n}{\eps^2(1-p_{\maxx})^2}\right)$ samples for each $Z_i$.
Putting everything together, we obtain the following theorem.

\begin{theorem}  \label{thm:bireach}
  There is an FPRAS for \bireach.
  The expected running time is $O\left(\eps^{-2}(1-p_{\maxx})^{-3}m^2n^3\right)$ for an $(1\pm\eps)$-approximation.
\end{theorem}

A natural question is what if $1-p_{\maxx}$ is close to $0$.
We can answer this in the case of uniform failure probabilities, i.e., when $p_e=p$ for all $e\in A$.
Let $p^*=1-1/(3m)$.  If $p\leq p^*$ then run Algorithm~\ref{alg:cluster-popping} as usual, to produce a sample in time 
$O(m^3n)$.  Otherwise, run the algorithm with modified weights $p_e=p^*$ for all $e\in A$, and let the output be~$S$,
so that $(V,S)$ is a root-connected subgraph.  Suppose $n-1\leq k<m$.  
Note that any root-connected subgraph with $k$ edges can be augmented 
to one with $k+1$ edges in at most $m$ ways.  Also, any root-connected subgraph with $k+1$ edges can be  
obtained by augmentation from at least one with $k$~edges.  Thus,
$$
\Pr(|S|=k+1)\leq \frac{m(1-p^*)}{p^*}\Pr(|S|=k)< \tfrac12\Pr(|S|=k). 
$$  
It follows that $\Pr(|S|=n-1)\geq\frac12$, i.e., there is a significant probability of observing 
an arborescence or directed tree.  Of course, the output distribution is not quite the one we want, but we can deal 
issue with that by (usual) rejection sampling:  simply retain $S$ with probability 
$$\left[\frac{p^*(1-p)}{p(1-p^*)}\right]^{|S|-n+1},$$ 
and otherwise 
run Algorithm~\ref{alg:cluster-popping} again to produce a fresh sample.  Note that the rejection probability 
is at most~$\frac12$, so the expected overall running time is still $O(m^3n)$.  This deals with exact sampling.
Since the number of arborescences can be computed exactly in polynomial time, an FPRAS for the case
$p>p^*$ follows easily.  It is not clear whether this method can be adapted to varying failure probabilities~$\mathbf{p}$.

\section{Coupling between reliability and bi-directed reachability} \label{sec:equivalence}

In this section, we give an alternative proof of Ball's equivalence between \relia\ and \bireach\ \cite[Corollary 1]{Ball80}.
Our proof constructs a coupling, between the (edge-weighted) distribution of connected subgraphs in the undirected setting,
and the (edge-weighted) distribution of root-connected subgraphs in the bi-directed setting.
This coupling, together with \cref{alg:cluster-popping}, yields an efficient exact sampler for connected subgraphs. 

We use $\{u,v\}$ to denote an undirected edge, and $(u,v)$ or $(v,u)$ to denote a directed one (namely an arc).
Let $G=(V,E)$ be an undirected graph, and $\mathbf{p}=(p_e)_{e\in E}$ be a vector of failure probabilities.
Let $\overrightarrow{G}=(V,A)$ be the bi-directed graph obtained by replacing every edge in $G$ with a pair of anti-parallel arcs.
Namely, $A=\{(u,v),(v,u)\mid\{u,v\}\in E\}$.
Moreover, let $p_{(u,v)}=p_{(v,u)}=p_{\{u,v\}}$ and denote these failure probabilities by $\*p'$.
For $S\subseteq E$ (or $S\subseteq A$), let $\wt(S):=\prod_{e\in S}(1-p_e)\prod_{e\in E\setminus S}p_e$ (or $\wt(S):=\prod_{e\in S}(1-p_e)\prod_{e\in A\setminus S}p_e$).

Consider the following coupling between the product distribution over edges of $G$ and the one over arcs of $\overrightarrow{G}$.
We reveal edges in a breadth-first search (BFS) fashion in both graphs, from the same ``root'' vertex $r$.
If an edge $\{u,v\}$ is present in the subgraph of $G$, we couple it with the arc $(u,v)$ or $(v,u)$, 
whose direction is pointing towards $r$ in the subgraph of $\overrightarrow{G}$.
The arc in the other direction is drawn independently from everything else.
The key observation is that to decide the set of vertices that can reach $r$, at any point,
only one direction of a bi-directed edge is useful and the other is irrelevant.
One can verify that in the end, the subgraph of~$G$ is connected if and only if the subgraph of $\overrightarrow{G}$ is root-connected.
We will formalize this intuition next.

Fix an arbitrary ordering of $V$, which will be used for the exploration, and let the first vertex be a distinguished root $r$.
Let $\+P(S)$ denote the power set of $S$ for a set $S$.
Define a mapping $\Phi:\+P(E)\rightarrow \+P(A)$ as follows.
For $S\subseteq E$, we explore all vertices that can reach $r$ in $(V,S)$ in a deterministic order,
and add arcs to $\Phi(S)$ in the direction towards $r$.
To be more specific, we maintain the set of explored and the set of active vertices, denoted by $V_{e}$ and $V_{a}$, respectively.
At the beginning, $V_e=\emptyset$ and $V_a=\{r\}$.
Given $V_e$ and $V_a$, let $v$ be the first vertex (according to the predetermined ordering) in $V_a$.
For all $u\in V\setminus V_e$, if $\{u,v\}\in S$, add $(u,v)$ to $\Phi(S)$ and add $u$ to $V_a$ ($u$ may be in $V_a$ already).
Then move $v$ from $V_a$ to $V_e$.
This process ends when all vertices that can reach $r$ in $(V,S)$ are explored.
Let $\sigma_S$ be the arriving order of $V_e$.
We will call $\sigma_S$ the traversal order.
We remark that if $\{u,v\}\in S$ then exactly one of the arcs $(u,v)$ and $(v,u)$ is in $\Phi(S)$,
and otherwise neither arc is in~$\Phi(S)$.

Strictly speaking, the exploration above is not a BFS ($V_a$ may contain a newly added vertex that is lower in the predetermined ordering than all other older vertices). 
To perform a BFS we need to in addition maintain a layer ordering, which seems unnecessary.
The key properties of the exploration are: 1) all edges incident to the current vertex are processed together, as a group; 
2) $V_e$ is always connected (or root-connected for $\Psi$ below).

Similarly, define $\Psi:\+P(A)\rightarrow\+P(E)$ as follows.
For $S'\subseteq A$, we again maintain $V_{e}$ and $V_{a}$, 
and initialize $V_e=\emptyset$ and $V_a=\{r\}$.
Given $V_e$ and $V_a$, let $v$ be the first vertex in $V_a$.
For all $u\in V\setminus V_e$, if $(u,v)\in S'$, add $\{u,v\}$ to $\Psi(S')$ and add $u$ to $V_a$.
Then move $v$ from $V_a$ to $V_e$.
This process ends when all vertices that can reach $r$ in $(V,S')$ are explored.
Analogously, let $\sigma_{S'}$ be the arriving order of $V_e$.
We remark that if $(u,v)\not\in S'$, and $v$ is visited before $u$, then $\{u,v\}\not\in\Psi(S')$, even in case of $(v,u)\in S'$.

Let $\Omega:=\{S\subseteq E\mid (V,S)\text{ is connected}\}$, and correspondingly $\overrightarrow{\Omega}:=\{S\subseteq A\mid (V,S)\text{ is root-connected}\}$.
We have the following lemma.

\begin{lemma}\label{lem:mapping}
  Let $\Phi$, $\Psi$, $\Omega$, and $\overrightarrow{\Omega}$ be defined as above.
  Then the following holds:
  \begin{enumerate}
    \item \label{item:into} if $S\in\Omega$, then $\Phi(S)\in\overrightarrow{\Omega}$;
    \item \label{item:into2} if $S'\in\overrightarrow{\Omega}$, then $\Psi(S')\in\Omega$;
    \item \label{item:psi-phi} if $S\in\Omega$, then $\Psi(\Phi(S))=S$;
    \item \label{item:surjective} $\Psi(\overrightarrow{\Omega}) = \Omega$;
    \item \label{item:weight} for any $S\in\Omega$, 
      \begin{align*}
        \wt(S)=\sum_{S'\in\Psi^{-1}(S)}\wt(S').
      \end{align*}
  \end{enumerate}
\end{lemma}
\begin{proof}
  \begin{enumerate}
    \item It is easy to verify that, at any point of the construction of $\Phi$,
      all vertices in $V_e$ can reach $r$, in both $(V,S)$ and $(V,\Phi(S))$.
      If $S\in\Omega$, then $V_e=V$ at the end of $\Phi$.
      Hence $(V,\Phi(S))$ is root-connected, and $\Phi(S)\in\overrightarrow{\Omega}$.
    \item This item is completely analogous to item \eqref{item:into}.
    \item If $\{u,v\}\in S$ and $u$ is processed first during the exploration, then $(v,u)\in\Phi(S)$.
      The traversal orderings $\sigma_S$ and $\sigma_{\Phi(S)}$ are the same.
      Hence, during the construction of $\Psi(\Phi(S))$, $u$ is still processed first, and $\{v,u\}\in\Psi(\Phi(S))$.
      On the other hand, if $\{u,v\}\not \in S$, then neither $(u,v)$ nor $(v,u)$ is in $\Phi(S)$ and thus $\{u,v\}\not\in\Psi(\Phi(S))$.
    \item This item is a straightforward consequence of items \eqref{item:into}, \eqref{item:into2}, and \eqref{item:psi-phi}.
    \item By item \eqref{item:psi-phi}, we have that $\Phi(S)\in\Psi^{-1}(S)$.
      Let 
      \begin{align*}
        \Phi_c(S):=\left\{(u,v)\mid(u,v)\not\in\Phi(S)\text{ and $v<u$ in the traversal order $\sigma_{\Phi(S)}$}\right\}.
      \end{align*}
      Note that $\Phi(S)\cup\Phi_c(S)$ covers all unordered pairs of vertices as $S\in\Omega$.
      Moreover,
      \begin{align}
        \label{eqn:weights}
        \prod_{e\in \Phi(S)}(1-p_e)\prod_{e\in\Phi_c(S)}p_e & = \wt(S).        
      \end{align}
      Call $S'$ consistent with $\Phi(S)$ if $\Phi(S)\subseteq S'$ and $S'\cap\Phi_c(S)=\emptyset$.

      We claim that $S'\in\Psi^{-1}(S)$ if and only if $S'$ is consistent with $\Phi(S)$.
      Suppose $S'$ is not consistent with $\Phi(S)$.
      Consider the exploration of $\Phi(S)$ and $S'$ in the construction of $\Psi$ simultaneously.
      Since $S'$ is not consistent with $\Phi(S)$, either $\Phi(S)\setminus S'\neq\emptyset$ or $S'\cap\Phi_c(S)\neq\emptyset$.
      Let $v$ be the first vertex during the exploration so that there is an arc $(u,v)\in\Phi(S)\setminus S'$, or $(u,v)\in S'\cap\Phi_c(S)$ for some $u\not\in V_e$.
      Since $S\in\Omega$, all vertices will be processed, and such a $v$ must exist.
      (In the latter case, since $(u,v)\in\Phi_c(S)$, $v$ is active first.)
      If $(u,v)\in\Phi(S)\setminus S'$, then $\{u,v\}\not\in\Psi(S')$ but $\{u,v\}\in \Psi(\Phi(S))$.
      If $(u,v)\in S'\cap\Phi_c(S)$, $\{u,v\}\not\in\Psi(\Phi(S))$ but $\{u,v\}\in\Psi(S')$.
      In either case, $\Psi(S')\neq\Psi(\Phi(S))=S$ (by item \eqref{item:psi-phi}).

      On the other hand, if $\Phi(S)\subseteq S'$ and $S'\cap\Phi_c(S)=\emptyset$,
      then we can trace through the construction of $\Psi(\Phi(S))$ and $\Psi(S')$ to verify that $\Psi(S')=\Psi(\Phi(S))=S$.

      The claim together with \eqref{eqn:weights} implies that
      \begin{align*}
        \sum_{S'\in\Psi^{-1}(S)}\wt(S') & = \sum_{S'\text{ is consistent with } \Phi(S)} \wt(S') \\
        & = \prod_{e\in \Phi(S)}(1-p_e)\prod_{e\in\Phi_c(S)}p_e = \wt(S). 
      \end{align*}
  \end{enumerate}
\end{proof}

\begin{lemma}\label{lem:coupling}
  $\Zrel(G;\*p)=\Zreach(\overrightarrow{G},r;\*p')$.
\end{lemma}
\begin{proof}
  First notice that 
  \begin{align*}
    \Zrel(G;\*p) = \sum_{S\in\Omega}\wt(S)
  \end{align*}
  and 
  \begin{align*}
    \Zreach(\overrightarrow{G},r;\*p') = \sum_{S\in\overrightarrow{\Omega}}\wt(S).
  \end{align*}
  By item \eqref{item:surjective} of \cref{lem:mapping}, $\Psi(\overrightarrow{\Omega}) = \Omega$, 
  implying that $\left( \Psi^{-1}(S) \right)_{S\in\Omega}$ is a partition of $\overrightarrow{\Omega}$.
  Combining this with item \eqref{item:weight} of \cref{lem:mapping},
  \begin{align*}
    \sum_{S\in\Omega}\wt(S) & = \sum_{S\in\Omega}\sum_{S'\in\Psi^{-1}(S)}\wt(S')\\
    & = \sum_{S'\in\overrightarrow{\Omega}}\wt(S').
  \end{align*}
  The lemma follows.
\end{proof}

\cref{lem:coupling} is first shown by Ball \cite[Corollary 2]{Ball80} via modifying edges one by one.
Instead, our proof is essentially a coupling argument and has a new consequence 
that \cref{alg:cluster-popping} can be used to sample edge-weighted connected subgraphs.
Recall our notation $\pi_G(\cdot)$, and generalise it to undirected graphs.  Thus, for an undirected (or directed) graph~$G$,
$\pi_G(\cdot)$ is the distribution resulting from drawing each edge (or arc) $e$ independently with probability $1-p_e$, 
and conditioning on the graph drawn being connected (or root-connected).

\begin{lemma}\label{lem:sampling}
  If a random root-connected subgraph $S'$ is drawn from $\pi_{\overrightarrow{G}}(\cdot)$,
  then $\Psi(S')$ has distribution $\pi_{G}(\cdot)$.
\end{lemma}
\begin{proof}
  Since $S'\in\overrightarrow{\Omega}$, by item \eqref{item:into2} of \Cref{lem:mapping}, $\Psi(S')\in\Omega$.
  Moreover, for any $s\in\Omega$,
  \begin{align*}
    \Pr[\Psi(S')=s] & = \sum_{s'\in\Psi^{-1}(s)}\Pr[S'=s'] \\
    & = \sum_{s'\in\Psi^{-1}(s)}\frac{\wt(s')}{\Zreach(\overrightarrow{G},r;\*p')}\\
    & = \frac{\wt(s)}{\Zrel(G;\*p)} = \pi_{G}(s),
  \end{align*}
  where we used item \eqref{item:weight} of \Cref{lem:mapping} and \Cref{lem:coupling} in the last line.
\end{proof}

There is also a coupling going the reversed direction of \cref{lem:sampling},
by drawing a random connected subgraph $S$ from $\pi_{G}(\cdot)$, 
mapping it to $\Phi(S)$, and excluding all arcs in $\Phi_c(S)$.
All other arcs are drawn independently.
The resulting $S'$ has distribution $\pi_{\overrightarrow{G}}(\cdot)$.
Its correctness is not hard to prove, given \cref{lem:mapping}, but it is not the direction of use to us and we omit its proof.

\Cref{thm:bireach} and \Cref{lem:coupling} imply the counting part of \Cref{thm:main}.
\Cref{thm:fpaus} and \Cref{lem:sampling} imply the sampling part of \Cref{thm:main}.

\section{Counting connected subgraphs of a specified cardinality}\label{sec:fixedsize}

In this section, we show that the sampling algorithm in \Cref{thm:main} also leads to an FPRAS for the number of connected subgraphs of any fixed size.

For a connected (undirected) graph $G=(V,E)$, as usual let $n=\abs{V}$ and $m=\abs{E}$.
For $n-1\le t\le m$, let $H_t\subset E$ be the set of connected (and spanning) subgraphs of size $t$, and $N_t=\abs{H_t}$.
Notice that $N_m=1$, and $N_{n-1}$ is the number of spanning trees, which can be computed in polynomial time exactly due to Kirchhoff's matrix-tree theorem.

The complements of connected subgraphs are independent sets of the co-graphic matroid associated with $G$,
and co-graphic matroids are representable \cite{Oxl92}.
Hence, by a breakthrough result of Huh and Katz \cite{HK12} (see also \cite{Lenz13} for a detailed derivation),
$(N_t)_t$ is a log-concave sequence.

\begin{proposition}\label{prop:log-concave}
  For any $n\le t\le m-1$,
  \begin{align*}
    N_{t-1}N_{t+1}\le N_t^2.
  \end{align*}
\end{proposition}

We remark that log-concavity of such a sequence has now been established for all matroids \cite{AHK18},
but here we only need the case of representable matroids.

Once we have the log-concavity and the sampling algorithm in \Cref{thm:main},
we can apply a technique of Jerrum and Sinclair \cite[Section 5]{JS89} to efficiently approximate $N_t$ for any $n-1\le t\le m$.

\begin{theorem}\label{thm:fixed-size}
  For any $n-1\le t\le m$, there is an FPRAS for $N_t$.
\end{theorem}

Here we sketch the outline of the algorithm.
We will only consider a uniform failure probability $p$ over all edges in the following.   
Also, we make no attempt to optimise the exponent in the polynomial running time.
The basic idea is to tune $p$ in the sampler of \Cref{thm:main} so that connected subgraphs of the desired size show up frequently enough.
First notice that \Cref{prop:log-concave} implies that the ratios $\frac{N_{t-1}}{N_t}$ is monotonically increasing.
It is straightforward to see that
\begin{align*}
  \frac{N_{n-1}}{N_{n}} \ge \frac{1}{m}, \quad \quad \text{ and } \quad\quad \frac{N_{m-1}}{N_m}\le m.
\end{align*}
Let $r_t=\frac{N_{t-1}}{N_t}$.
Hence,
\begin{align}\label{eqn:ratio-size-t}
  \frac{1}{m}\le r_n\le r_{n+1} \le \dots \le r_m\le{m}.
\end{align}

We will use $r=\frac{1-p}{p}$ to denote the edge weight when the failure probability of an edge is $p$.
With a little abuse of notation, let $\pi_{r}(\cdot)$ be the distribution over connected subgraphs when each edge is removed with probability $p=\frac{1}{1+r}$ independently.
(So $\pi_{r}(\cdot)$ is a product distribution on the edges, conditioned on the result being connected.) 
It is easy to see that for a connected subgraph $R\subset E$, $\pi_r(R)\propto r^{\abs{R}}$.
We note that $\pi_{r_t}(H_{t-1}) = \pi_{r_t}(H_{t})$,
and for any $i < t$,
\begin{align*}
  \frac{\pi_{r_t}(H_t)}{\pi_{r_t}(H_{i})}= r_t^{t-i}\cdot\frac{N_t}{N_i} 
  = r_t^{t-i}\cdot\prod_{j=i+1}^{t}\frac{N_{j}}{N_{j-1}}
  = r_t^{t-i}\cdot\prod_{j=i+1}^{t}r_{j}^{-1}
  \ge r_t^{t-i}r_t^{i-t} =1,
\end{align*}
where we used \eqref{eqn:ratio-size-t}.
Similarly, for any $i>t$, $\pi_{r_t}(H_t)\ge\pi_{r_t}(H_i)$.
Note that $\sum_{i=n-1}^m\pi_{r_t}(H_i)=1$. We conclude that
\begin{align}  \label{eqn:size-t-lb}
  \pi_{r_t}(H_{t-1})=\pi_{r_t}(H_t)\ge\frac{1}{m}.
\end{align}
Thus, if we run the sampling algorithm of \Cref{thm:main} with $p_t=\frac{1}{1+r_t}$, 
there is a significant probability to see subgraphs in $H_t$.

To utilise the argument above, we need to know $r_t$.
This can be done inductively, since
\begin{align*}
  \pi_{r_t}(H_{t-2})=\frac{r_{t-1}}{r_t}\cdot\pi_{r_t}(H_{t-1})\ge\frac{1}{m^3},
\end{align*}
where we used \eqref{eqn:ratio-size-t} and \eqref{eqn:size-t-lb}.
Rewrite $N_t$ as
\begin{align*}
  N_t=N_{m}\cdot\prod_{i=m}^{t+1}\frac{N_{i-1}}{N_i} = \prod_{i=m}^{t+1} r_i,
\end{align*}
and our estimator of $N_t$ will be the product of estimators for $r_i$ where $i\in[t+1,m]$.
A complete description is given in \Cref{alg:fix-size-counting}.
We should set $T$ to be a sufficiently large number (but still polynomial in $n$) so that the variances of the estimators are small enough.
Notice that $N_{m-1}$ is easy to compute since it is just the number of edges in $G$ that are not bridges.

\begin{algorithm}
  \caption{Approximately count connected subgraphs of a fixed size $t\in[n, m-2]$}
  \label{alg:fix-size-counting}
  \begin{algorithmic}
    \STATE {Let $\widetilde{r}\gets\frac{N_{m-1}}{N_m}$ and $\widetilde{N}=N_{m-1}$.}
      \FOR{$i=m-2,m-1,\dots,t$}
        \IF {$\widetilde{r}\not\in[1/2m,2m]$}
        \RETURN {0} \COMMENT {Note the bounds in \eqref{eqn:ratio-size-t}}
        \ENDIF
        \STATE {Draw $T$ samples from $\pi_{\widetilde{r}}(\cdot)$ using \Cref{alg:cluster-popping}, yielding a set $Y$.}
        \IF {$\abs{Y\cap H_i}=0$ or $\abs{Y\cap H_{i+1}}=0$}
          \RETURN {0}
        \ENDIF        
        \STATE {Let $\widetilde{r}\gets \widetilde{r}\cdot\frac{\abs{Y\cap H_i}}{\abs{Y\cap H_{i+1}}}$ and $\widetilde{N}\gets\widetilde{N}/\widetilde{r}$.}
      \ENDFOR
      \RETURN {$\widetilde{N}$}
  \end{algorithmic}
\end{algorithm}

The analysis of \Cref{alg:fix-size-counting} is identical to the proof of \cite[Theorem 5.3]{JS89} and thus omitted.
\Cref{thm:fixed-size} is a direct consequence of \Cref{alg:fix-size-counting}.

\section{Concluding remarks}\label{sec:conclusion}

In this paper we give an FPRAS for \relia\ (or, equivalently, \bireach), by confirming a conjecture of Gorodezky and Pak \cite{GP14}.
We also give an exact sampler for edge-weighted connected subgraphs with polynomial running time in expectation.
The core ingredient of our algorithms is the cluster-popping algorithm to sample root-connected subgraphs, namely \cref{alg:cluster-popping}.
We manage to analyze it using the \prs\ framework.

\relia\ is equivalent to counting weighted connected subgraphs,
which is the evaluation of the Tutte polynomial $T_{G}(x,y)$ for points $x=1$ and $y>1$.
An interesting question is about the dual of this half-line, namely for points $x>1$ and $y=1$,
whose evaluation is to count weighted acyclic subgraphs.
It is well known that for a planar graph $G$, $T_{G}(x,1)=T_{G^{\ast}}(1,x)$ where $G^{\ast}$ is the planar dual of $G$ \cite{Oxl92}.
Hence, \cref{thm:main} implies that in planar graphs, 
$T_{G}(x,1)$ can be efficiently approximated for $x>1$.
Can we remove the restriction of planar graphs?

Another interesting direction is to generalize \Cref{alg:cluster-popping} beyond bi-directed graphs.
What about Eulerian graphs?
Is approximating \reach\ \NP-hard in general?

\section*{Acknowledgements}

We thank Mark Huber for bringing reference \cite{GP14} to our attention,
Mark Walters for the coupling idea leading to \cref{lem:coupling},
David Harris for proposing the question answered in \Cref{sec:fixedsize},
and Igor Pak for comments on an earlier version.
We also thank the organizers of the ``LMS -- EPSRC Durham Symposium on Markov Processes, Mixing Times and Cutoff'',
where part of the work is carried out.

\bibliographystyle{alpha}
\bibliography{PRS}

\begin{thebibliography}{JVW90}

\bibitem[AHK18]{AHK18}
Karim Adiprasito, June Huh, and Eric Katz.
\newblock Hodge theory for combinatorial geometries.
\newblock {\em Ann. of Math. (2)}, 188(2):381--452, 2018.

\bibitem[Bal80]{Ball80}
Michael~O. Ball.
\newblock Complexity of network reliability computations.
\newblock {\em Networks}, 10(2):153--165, 1980.

\bibitem[Bal86]{Ball86}
Michael~O. Ball.
\newblock Computational complexity of network reliability analysis: An
  overview.
\newblock {\em IEEE Trans. Rel.}, 35(3):230--239, 1986.

\bibitem[BP83]{BP83}
Michael~O. Ball and J.~Scott Provan.
\newblock Calculating bounds on reachability and connectedness in stochastic
  networks.
\newblock {\em Networks}, 13(2):253--278, 1983.

\bibitem[Col87]{Col87}
Charles~J. Colbourn.
\newblock {\em The Combinatorics of Network Reliability}.
\newblock Oxford University Press, 1987.

\bibitem[CPP02]{CPP02}
Henry Cohn, Robin Pemantle, and James~G. Propp.
\newblock Generating a random sink-free orientation in quadratic time.
\newblock {\em Electron. J. Combin.}, 9(1):10:1--10:13, 2002.

\bibitem[GJ08]{GJ08}
Leslie~Ann Goldberg and Mark Jerrum.
\newblock Inapproximability of the {T}utte polynomial.
\newblock {\em Inf. Comput.}, 206(7):908--929, 2008.

\bibitem[GJ14]{GJ14}
Leslie~Ann Goldberg and Mark Jerrum.
\newblock The complexity of computing the sign of the {T}utte polynomial.
\newblock {\em {SIAM} J. Comput.}, 43(6):1921--1952, 2014.

\bibitem[GJL17]{GJL17}
Heng Guo, Mark Jerrum, and Jingcheng Liu.
\newblock Uniform sampling through the {L}ovasz local lemma.
\newblock In {\em {STOC}}, pages 342--355, 2017.

\bibitem[GP14]{GP14}
Igor Gorodezky and Igor Pak.
\newblock Generalized loop-erased random walks and approximate reachability.
\newblock {\em Random Struct. Algorithms}, 44(2):201--223, 2014.

\bibitem[Hag91]{Hag91}
Jane~N. Hagstrom.
\newblock Computing rooted communication reliability in an almost acyclic
  digraph.
\newblock {\em Networks}, 21(5):581--593, 1991.

\bibitem[HK12]{HK12}
June Huh and Eric Katz.
\newblock Log-concavity of characteristic polynomials and the {B}ergman fan of
  matroids.
\newblock {\em Math. Ann.}, 354(3):1103--1116, 2012.

\bibitem[HS14]{HS14}
David~G. Harris and Aravind Srinivasan.
\newblock Improved bounds and algorithms for graph cuts and network
  reliability.
\newblock In {\em {SODA}}, pages 259--278. {SIAM}, 2014.

\bibitem[Jer81]{Jer81}
Mark Jerrum.
\newblock On the complexity of evaluating multivariate polynomials.
  \emph{Ph.D.\ dissertation}.
\newblock Technical Report CST-11-81, Dept. Comput. Sci., Univ.\ Edinburgh,
  1981.

\bibitem[JS89]{JS89}
Mark Jerrum and Alistair Sinclair.
\newblock Approximating the permanent.
\newblock {\em {SIAM} J. Comput.}, 18(6):1149--1178, 1989.

\bibitem[JS93]{JS93}
Mark Jerrum and Alistair Sinclair.
\newblock Polynomial-time approximation algorithms for the {I}sing model.
\newblock {\em {SIAM} J. Comput.}, 22(5):1087--1116, 1993.

\bibitem[JVW90]{JVW90}
Fran\c{c}ois Jaeger, Dirk~L. Vertigan, and Dominic J.~A. Welsh.
\newblock On the computational complexity of the {J}ones and {T}utte
  polynomials.
\newblock {\em Math. Proc. Cambridge Philos. Soc.}, 108(1):35--53, 1990.

\bibitem[Kar99]{Kar99}
David~R. Karger.
\newblock A randomized fully polynomial time approximation scheme for the
  all-terminal network reliability problem.
\newblock {\em {SIAM} J. Comput.}, 29(2):492--514, 1999.

\bibitem[Kar16]{Kar16}
David~R. Karger.
\newblock A fast and simple unbiased estimator for network (un)reliability.
\newblock In {\em {FOCS}}, pages 635--644, 2016.

\bibitem[Kar17]{Kar17}
David~R. Karger.
\newblock Faster (and still pretty simple) unbiased estimators for network
  (un)reliability.
\newblock In {\em {FOCS}}, pages 755--766, 2017.

\bibitem[KL85]{KL85}
Richard~M. Karp and Michael Luby.
\newblock Monte-{C}arlo algorithms for the planar multiterminal network
  reliability problem.
\newblock {\em J. Complexity}, 1(1):45--64, 1985.

\bibitem[KS11]{KS11}
Kashyap Babu~Rao Kolipaka and Mario Szegedy.
\newblock {M}oser and {T}ardos meet {L}ov{\'{a}}sz.
\newblock In {\em STOC}, pages 235--244, 2011.

\bibitem[Len13]{Lenz13}
Matthias Lenz.
\newblock The {$f$}-vector of a representable-matroid complex is log-concave.
\newblock {\em Adv. in Appl. Math.}, 51(5):543--545, 2013.

\bibitem[MT10]{MT10}
Robin~A. Moser and G{\'{a}}bor Tardos.
\newblock A constructive proof of the general {L}ov{\'{a}}sz {L}ocal {L}emma.
\newblock {\em J. {ACM}}, 57(2), 2010.

\bibitem[Oxl92]{Oxl92}
James~G. Oxley.
\newblock {\em Matroid theory}.
\newblock Oxford University Press, 1992.

\bibitem[PB83]{PB83}
J.~Scott Provan and Michael~O. Ball.
\newblock The complexity of counting cuts and of computing the probability that
  a graph is connected.
\newblock {\em {SIAM} J. Comput.}, 12(4):777--788, 1983.

\bibitem[She85]{Shearer85}
James~B. Shearer.
\newblock On a problem of {S}pencer.
\newblock {\em Combinatorica}, 5(3):241--245, 1985.

\bibitem[Val79]{Val79}
Leslie~G. Valiant.
\newblock The complexity of enumeration and reliability problems.
\newblock {\em {SIAM} J. Comput.}, 8(3):410--421, 1979.

\bibitem[VW92]{VW92}
Dirk Vertigan and Dominic J.~A. Welsh.
\newblock The compunational complexity of the {T}utte plane: the bipartite
  case.
\newblock {\em Combin. Probab. Comput.}, 1:181--187, 1992.

\bibitem[Wel93]{Wel93}
Dominic J.~A. Welsh.
\newblock {\em Complexity: knots, colourings and counting}, volume 186 of {\em
  London Mathematical Society Lecture Note Series}.
\newblock Cambridge University Press, 1993.

\bibitem[Wel99]{Wel99}
Dominic J.~A. Welsh.
\newblock The {T}utte polynomial.
\newblock {\em Random Struct. Algorithms}, 15(3-4):210--228, 1999.

\bibitem[Wil96]{Wilson96}
David~B. Wilson.
\newblock Generating random spanning trees more quickly than the cover time.
\newblock In {\em STOC}, pages 296--303, 1996.

\end{thebibliography}

\end{document}